\documentclass{eptcs}
\providecommand{\event}{}
\providecommand{\volume}{}
\providecommand{\anno}{2009}
\providecommand{\firstpage}{1}

\usepackage{breakurl,amssymb,amsmath,amsthm,verbatim,rotating,tabularx}
\newtheorem{algorithm}{Algorithm}{\bfseries}{}
\newtheorem{definition}{Definition}{\bfseries}{}
\newtheorem{theorem}{Theorem}{\bfseries}{}
\newtheorem{proposition}{Proposition}{\bfseries}{}
\newtheorem{lemma}{Lemma}{\bfseries}{}
\newtheorem{corollary}{Corollary}{\bfseries}{}

\title{Complexity of Strong Implementability}
\author{Clemens Thielen and Sven O. Krumke
\institute{Department of Mathematics, University of Kaiserslautern,
Paul-Ehrlich-Str.~14, D-67663~Kaiserslautern, Germany}
\email{\{thielen,krumke\}@mathematik.uni-kl.de}}

\def\titlerunning{Complexity of Strong Implementability}
\def\authorrunning{C. Thielen \& S. O. Krumke}

\begin{document}
\maketitle

\begin{abstract}
  We consider the question of implementability of a social choice
  function in a classical setting where the preferences of finitely
  many selfish individuals with private information have to be
  aggregated towards a social choice. This is one of the central
  questions in mechanism design. If the concept of \emph{weak}
  implementation is considered, the \emph{Revelation Principle} states
  that one can restrict attention to \emph{truthful implementations}
  and \emph{direct revelation mechanisms}, which implies that
  implementability of a social choice function is ``easy'' to check.
  For the concept of \emph{strong implementation}, however, the
  Revelation Principle becomes invalid, and the complexity of deciding
  whether a given social choice function is strongly implementable has
  been open so far. In this paper, we show by using methods from polyhedral
  theory that strong implementability of a social choice function can be decided in
  polynomial space and that each of the payments needed for strong
  implementation can always be chosen to be of polynomial encoding
  length. Moreover, we show that strong implementability of a social
  choice function involving only a single selfish individual can be
  decided in polynomial time via linear programming.
\end{abstract}

\section{Introduction}

Mechanism design is a classical area of noncooperative game
theory~\cite{Osborne+Rubinstein:book} and
microeconomics~\cite{Mas-Colell+etal:book} which studies how privately
known preferences of several people can be aggregated towards a social
choice. Applications include the design of voting procedures,
the writing of contracts among parties, and the construction of
procedures for deciding upon public projects.  Recently, the study of
the Internet has fostered the interest in algorithmic aspects of
mechanism design \cite{Nisan+Ronen:mech-design}.

In the classical social choice setting considered in this paper, there
are $n$ selfish agents, which must make a collective decision from
some finite set~$X$ of possible social choices. Each agent~$i$ has a
private value~$\theta_i\in\Theta_i$ (called the agent's \emph{type}),
which influences the preferences of all agents over the alternatives
in $X$. Formally, this is modeled by a \emph{valuation
  function}~$V_i:X\times\Theta \rightarrow \mathbb{Q}$ for each
agent~$i$, where $\Theta=\Theta_1\times\dots\times\Theta_n$.  Every
agent~$i$ reports some information~$s_i$ from a set~$S_i$ of possible
\emph{bids} of $i$ to the mechanism designer who must then choose an
alternative from $X$ based on these bids.  The goal of the mechanism
designer is to implement a given social choice function~$f:\Theta
\rightarrow X$, that is, to make sure that the alternative~$f(\theta)$
is always chosen in equilibrium when the vector of true types is
$\theta=(\theta_1,\dots,\theta_n)$. To achieve this, the mechanism
designer hands out a payment~$P_i(\theta)$ to each agent~$i$,´ which
depends on the bids.  Each agent then tries to maximize the sum of her
valuation and payment by choosing an appropriate bid depending on her
type.  A \emph{mechanism}~$\Gamma=(S_1,\dots,S_n,g,P)$ is defined by
the sets $S_1,\dots,S_n$ of possible bids of the agents, an
\emph{outcome function}~$g:S_1\times\dots\times S_n \rightarrow X$,
and the payment scheme~$P=(P_1,\dots,P_n)$.

In the most common concept called \emph{weak}
implementation, the mechanism~$\Gamma$ is said to \emph{implement}
the social choice function~$f$ if \emph{some} (Bayesian) equilibrium of the
noncooperative game defined by the mechanism yields the outcome
specified by~$f$. An important result known
as the \emph{Revelation Principle}
(cf. \cite[p.~884]{Mas-Colell+etal:book}) states that a social choice
function is weakly implementable if and only if it can be
\emph{truthfully implemented} by a \emph{direct revelation mechanism},
which means that $f$ can be implemented by a mechanism with
$S_i=\Theta_i$ for all $i$ and truthful reporting as an equilibrium
that yields the outcome specified by $f$. As a result, the question
whether there exists a mechanism that weakly implements a given social
choice function~$f$ can be easily answered in time polynomial in $|\Theta|$
by checking for negative cycles in complete directed graphs on the
agents' type spaces with changes of valuations as edge weights
(cf.~\cite{Mueller+etal:weak-monotonicity}).

However, there is an obvious drawback in considering weak implementation:
Although a mechanism~$\Gamma$ may have \emph{some} equilibrium that yields the outcome
specified by~$f$, there may be \emph{other} equilibria that yield \emph{different}
outcomes. Thus, the concept of weak implementation heavily relies on the implicit assumption
that the agents always play the equilibrium that the mechanism designer wants if there
is more than one.

The standard way to avoid this problem is to consider the more robust concept of
implementation called \emph{strong implementation}. A mechanism~$\Gamma$ is said to
\emph{strongly implement} the social choice function~$f$ if
\emph{every} equilibrium of $\Gamma$ yields the outcome specified by~$f$.
For strong implementation, the Revelation Principle does not hold, so one cannot,
in general, restrict attention to direct revelation mechanisms and
truthful implementations when trying to decide whether a social choice
function is strongly implementable. In tackling the question whether a
given social choice function can be \emph{strongly} implemented, it is
not even a priori clear that one can restrict attention to finite
sets~$S_i$ or polynomially sized payments. Thus, to the best of our
knowledge, the complexity has been open so far. The main result
of this paper is that strong implementability of a social choice
function can be decided in polynomial space. In particular, if a
social choice function can be strongly implemented, our results show
that each of the payments in a mechanism that strongly implements it
can be chosen to be of polynomial encoding length. It seems unlikely
that the problem is contained in \textsf{NP}, at least the
characterizations of strong implementability developed so far require
an exponential number of (polynomially sized) certificates. We conjecture
that deciding strong implementability is in fact \textsf{PSPACE}-complete.
However, we show that the problem can be solved in polynomial time in case
of a single agent.

\paragraph{\bf{Problem Definition}}

We are given $n$ \emph{agents} identified with the
set~$N=\{1,\dots,n\}$ and a finite set~$X$ of possible \emph{social
  choices}.  For each agent~$i$, there is a finite set~$\Theta_i$ of
possible \emph{types} and we write
$\Theta=\Theta_1\times\dots\times\Theta_n$.  The true type~$\theta_i$
of agent~$i$ is known only to the agent herself. However, there is a
commonly known \emph{probability distribution}~$p:\Theta \rightarrow
\mathbb{Q}$ on $\Theta$ satisfying $p(\theta)\geq 0$ for all
$\theta\in\Theta$ and $\sum_{\theta\in\Theta}{p(\theta)}=1$. Thus,
every agent knows the probabilities of all possible
vectors~$\theta\in\Theta$ of types of all agents. Without loss of
generality, we assume the marginal probability~$p_i(\theta_i):=
\sum_{\theta_{-i}\in\Theta_{-i}}{p(\theta_{-i},\theta_i)}$ to be
strictly positive for every $\theta_i\in\Theta_i$. The beliefs of an
agent~$i$ of type~$\theta_i$ are then given by the (conditioned)
probability distribution~$q_i(.|\theta_i)$ over $\Theta_{-i}$ defined
by\small
\begin{equation*}
  q_i(\theta_{-i}|\theta_i):=\frac{p(\theta_{-i},\theta_i)}{p_i(\theta_i)}.
\end{equation*}\normalsize
All agents as well as the mechanism designer know these beliefs for
every $i$ and every $\theta_i$ since they know the probability
distribution~$p$. However, as mentioned above, only agent~$i$ knows
the true value of $\theta_i$.

Each agent~$i$ has a \emph{valuation function}~$V_i:X\times\Theta
\rightarrow \mathbb{Q}$, where $V_i(x,\theta)$ specifies the value
that agent~$i$ assigns to alternative~$x\in X$ when the types of the
agents are $\theta\in\Theta$. A \emph{social choice function} in this
setting is a function~$f:\Theta \rightarrow X$ that assigns an
alternative $f(\theta)\in X$ to every vector~$\theta$ of types.

\begin{definition}
  A \emph{mechanism}~$\Gamma=(S_1,\dots,S_n,g,P)$ consists of a
  set~$S_i$ of possible \emph{bids} for each agent~$i$, an
  \emph{outcome function}~$g:S \rightarrow X$ and a \emph{payment
    scheme}~$P:S \rightarrow \mathbb{Q}^n$, where
  $S:=S_1\times\dots\times S_n$.  $\Gamma$ is called a \emph{direct
    revelation mechanism} if $S_i=\Theta_i$ for all $i\in N$. We
  denote the direct revelation
  mechanism~$(\Theta_1,\dots,\Theta_n,f,P)$ defined by a social choice
  function~$f$ and a payment scheme~$P$ by $\Gamma_{(f,P)}$.

  A \emph{strategy} for agent~$i$ in the mechanism~$\Gamma$ is a
  function~$\alpha_i:\Theta_i \rightarrow S_i$ that defines a
  bid~$\alpha_i(\theta_i)\in S_i$ for every possible type $\theta_i$ of
  agent~$i$. A \emph{strategy profile} (in the mechanism~$\Gamma$) is
  an $n$-tuple~$\alpha=(\alpha_1,\dots\alpha_n)$ containing a strategy
  $\alpha_i$ for each agent~$i$.
\end{definition}

\begin{definition}\label{bayesianequ}
  Given a mechanism~$\Gamma=(S_1,\dots,S_n,g,P)$ and an
  $(n-1)$-tuple~$\alpha_{-i}$ of strategies for all agents except $i$,
  the \emph{expected utility} from a bid~$s_i\in S_i$ for agent~$i$
  when her type is $\theta_i$ is defined as
  \begin{equation*}\textstyle
    U^{\Gamma}_i(\alpha_{-i},s_i|\theta_i):= \sum_{\theta_{-i}\in\Theta_{-i}}
    {q_i(\theta_{-i}|\theta_i)
      \cdot\Big(V_i(g(\alpha_{-i}(\theta_{-i}),s_i),\theta)+P_i(\alpha_{-i}(\theta_{-i}),s_i)\Big)} 
	\end{equation*}
The strategy profile~$\alpha=(\alpha_1,\dots,\alpha_n)$ is a
\emph{(Bayesian) equilibrium} of $\Gamma$ if $\alpha_i(\theta_i)$ maximizes the
expected utility of an agent~$i$ of type~$\theta_i$ for every $i\in N$ and every
$\theta_i\in\Theta_i$, i.e., if
$\alpha_i(\theta_i)\in \operatorname{argmax}_{s_i\in S_i}{U^{\Gamma}_i(\alpha_{-i},s_i|\theta_i)}$ for all $i\in N,\theta_i\in\Theta_i$.
A direct revelation mechanism~$\Gamma_{(f,P)}$ is called \emph{incentive compatible} if truthful reporting is an equilibrium.
\end{definition}

Note that incentive compatibility of a mechanism is often defined as the property that truthful reporting is a \emph{dominant strategy equilibrium}.
However, since we consider \emph{Bayesian} equilibria, the definition given above is the natural choice in our setting.

\begin{definition}
  The mechanism~$\Gamma=(S_1,\dots,S_n,g,P)$ \emph{strongly
  implements} the social choice function~$f$ if
  \begin{enumerate}
  \item $\Gamma$ has at least one equilibrium,
  \item Every equilibrium~$\alpha$ of $\Gamma$ satisfies
    $g\circ\alpha=f$.
  \end{enumerate}
  The social choice function~$f$ is called \emph{strongly
  implementable} if there exists a mechanism~$\Gamma$ that strongly
  implements $f$.
\end{definition}

%With these definitions, the Strong Implementability Problem can be formulated as follows:

\begin{definition}[Strong Implementability Problem]\mbox{}\\
  \begin{tabularx}{\linewidth}{lX}
    INSTANCE: & The number~$n$ of agents, the set~$X$ of possible social choices, the sets~$\Theta_i$ of possible types of the
    agents, the valuation functions~$V_i:X\times\Theta \rightarrow \mathbb{Q}$, the probability distribution~$p$ over $\Theta$, and
    the social choice function~$f:\Theta \rightarrow X$. \\
    QUESTION: & Is $f$ strongly implementable? \
  \end{tabularx}
\end{definition}

To encode an instance of Strong Implementability, we need to do the
following: For every valuation function~$V_i:X\times\Theta
\rightarrow \mathbb{Q}$, we need to store $|X|\cdot|\Theta|$ rational
numbers. For the probability distribution~$p:\Theta\rightarrow
\mathbb{Q}$, we need to store $|\Theta|$ nonnegative rational
numbers. The social choice function~$f:\Theta \rightarrow X$ has
encoding length~$|\Theta|\cdot\log(|X|)$.
Thus, the encoding length of an instance of Strong Implementability
is in $\Omega(|X|\cdot|\Theta|\cdot n)$.

Observe that, in general, the encoding length is exponential in the number~$n$
of agents, since, even when each agent has just two possible types, we have $|\Theta|=2^n$.
It remains an interesting question for future research to study the problem and its complexity
if the valuation functions~$V_i$, the probability distribution~$p$, and the social choice function~$f$ are not explicitly specified
but can only be accessed via an oracle. In this case, one would be interested in the existence/non-existence
of oracle-polynomial time/space algorithms. We conjecture that exponential lower bounds can
be proved in this setting.

\section{Solvability in Polynomial Space}

In this section, we show that Strong Implementability is contained in
the complexity class~\textsf{PSPACE}.  One of the key ingredients of
our argumentation is a result due to Mookherjee and
Reichelstein~\cite{Mookherjee+Reichelstein:augmented-revelation}
called the \emph{Augmented Revelation Principle}. The mechanism design
setting considered
in~\cite{Mookherjee+Reichelstein:augmented-revelation} is similar to
ours. The proof of the Augmented Revelation Principle given
in~\cite{Mookherjee+Reichelstein:augmented-revelation} is focused on
the case where no payments are allowed in a mechanism. We now give a
proof for our setting of mechanisms \emph{with} payments.

\begin{definition}
  A mechanism~$\Gamma=(S_1,\dots,S_n,g,P)$ is called \emph{augmented
  revelation mechanism} if $S_i=\Theta_i\cup T_i$ for all $i\in N$
  and arbitrary sets~$T_i$.
\end{definition}

In the above definition, the elements of the set~$T_i$ represent
additional bids available to agent~$i$, in addition to her possible
types.

\begin{theorem}[The Augmented Revelation Principle]\label{augrefprinciple}\mbox{}\\
  If a social choice function~$f:\Theta \rightarrow X$ is strongly
  implementable, then $f$ can be strongly implemented by an augmented
  revelation mechanism in which truthful reporting is an equilibrium.
\end{theorem}

\begin{proof}
Let $\Gamma=(S_1,\dots,S_n,g,P)$ strongly implement $f$. We construct an augmented revelation
mechanism~$\bar{\Gamma}=(\bar{S}_1,\dots,\bar{S}_n,\bar{g},\bar{P})$ that strongly implements $f$ as in the original proof of the
Augmented Revelation Principle in \cite{Mookherjee+Reichelstein:augmented-revelation}. However, we additionally have to define the new payment
scheme~$\bar{P}$ in terms of the given payment scheme~$P$.
Given an arbitrary equilibrium~$\alpha=(\alpha_1,\dots,\alpha_n)$ of $\Gamma$, we define $\bar{S}_i:=\Theta_i\cup T_i$, where
\begin{displaymath}
T_i:=\{s_i\in S_i\,|\, s_i\notin \textrm{image}(\alpha_i)\}
\end{displaymath}
and $\textrm{image}(\alpha_i)=\{\alpha_i(\theta_i)\,|\,\theta_i\in\Theta_i\}$ denotes the image of the function~$\alpha_i:\Theta_i \rightarrow S_i$.
We consider the functions~$\phi_i:\bar{S}_i \rightarrow S_i$ given by
\begin{displaymath}
\phi_i(\bar{s_i}):=\left\{
	\begin{array}{cccc}
	\alpha_i(\theta_i) & \quad &  \textrm{if } \bar{s}_i=\theta_i \textrm{ for } \theta_i\in\Theta_i\\
	\bar{s}_i          & \quad &  \textrm{if } \bar{s}_i\in T_i\\
	\end{array}
	\right.
\end{displaymath}
and define the outcome function~$\bar{g}:\bar{S}\rightarrow X$ as $\bar{g}:=g\circ\phi$, where $\phi=(\phi_1,\dots,\phi_n)$.
The payment scheme~$\bar{P}:\bar{S} \rightarrow \mathbb{Q}$ is defined analogously as $\bar{P}:=P\circ\phi$.

To show that $\bar{\Gamma}$ strongly implements $f$, suppose that $\bar{\alpha}=(\bar{\alpha}_1,\dots,\bar{\alpha}_n)$ is
an equilibrium of $\bar{\Gamma}$ and consider the the strategy profile $\alpha^{\ast}=(\alpha^{\ast}_1,\dots,\alpha^{\ast}_n)$
in $\Gamma$ given by $\alpha^{\ast}_i:=\phi_i\circ\bar{\alpha}_i$.
By definition of $\alpha^{\ast}$, $\bar{g}$, and $\bar{P}$ we then have $g\circ\alpha^{\ast}=g\circ\phi\circ\bar{\alpha}=\bar{g}\circ\bar{\alpha}$
and $P\circ\alpha^{\ast}=P\circ\phi\circ\bar{\alpha}=\bar{P}\circ\bar{\alpha}$.
We claim that $\alpha^{\ast}$ is an equilibrium of $\Gamma$. Since every $\phi_i:\bar{S}_i \rightarrow S_i$ is surjective, we can choose, for
every fixed $s_i\in S_i$, an element~$\bar{s}_i\in \bar{S}_i$ with $\phi_i(\bar{s}_i)=s_i$.
Then, for every $i\in N$ and every possible $\theta_i\in\Theta_i$, we have
\begin{eqnarray}
	U^{\Gamma}_i(\alpha^{\ast}_{-i},\alpha^{\ast}_i(\theta_i)|\theta_i)
	&  =   & \sum_{\theta_{-i}\in\Theta_{-i}}{q_i(\theta_{-i}|\theta_i)
					 \biggl(V_i(g(\alpha^{\ast}(\theta)),\theta) + P_i(\alpha^{\ast}(\theta))\biggr)} \nonumber\\
	&  =   & \sum_{\theta_{-i}\in\Theta_{-i}}{q_i(\theta_{-i}|\theta_i)
					 \biggl(V_i(\bar{g}(\bar{\alpha}(\theta)),\theta) + \bar{P}_i(\bar{\alpha}(\theta))\biggr)} \nonumber\\
	& \geq & \sum_{\theta_{-i}\in\Theta_{-i}}{q_i(\theta_{-i}|\theta_i)
					 \biggl(V_i(\bar{g}(\bar{\alpha}_{-i}(\theta_{-i}),\bar{s}_i),\theta)
					 + \bar{P}_i(\bar{\alpha}_{-i}(\theta_{-i}),\bar{s}_i)\biggr)} \nonumber\\
	&  =   & \sum_{\theta_{-i}\in\Theta_{-i}}{q_i(\theta_{-i}|\theta_i)
					 \biggl(V_i(g(\alpha^{\ast}_{-i}(\theta_{-i}),\phi_i(\bar{s}_i)),\theta)
					 + P_i(\alpha^{\ast}_{-i}(\theta_{-i}),\phi_i(\bar{s}_i))\biggr)} \nonumber\\
	&  =   & \sum_{\theta_{-i}\in\Theta_{-i}}{q_i(\theta_{-i}|\theta_i)
					 \biggl(V_i(g(\alpha^{\ast}_{-i}(\theta_{-i}),s_i),\theta)
					 + P_i(\alpha^{\ast}_{-i}(\theta_{-i}),s_i)\biggr)} \nonumber\\
	&  =   & U^{\Gamma}_i(\alpha^{\ast}_{-i},s_i|\theta_i) \nonumber
\end{eqnarray}
where the inequality follows since $\bar{\alpha}$ is an equilibrium of $\bar{\Gamma}$. Since $s_i\in S_i$ was arbitrary, this shows that
$\alpha^{\ast}$ is an equilibrium of $\Gamma$ as claimed. So since $\Gamma$ strongly implements $f$, it follows that
$f=g\circ\alpha^{\ast}=\bar{g}\circ\bar{\alpha}$, i.e., the equilibrium $\bar{\alpha}$ yields the outcome specified by $f$.
Hence, it just remains to show that truthful bidding is an
equilibrium of $\bar{\Gamma}$. But this follows easily since, for every $\theta\in\Theta$, we have
$\bar{g}(\theta)=(g\circ\phi)(\theta)=g(\alpha(\theta))$ and $\bar{P}(\theta)=(P\circ\phi)(\theta)=P(\alpha(\theta))$ and $\alpha$
is an equilibrium of $\Gamma$.
\end{proof}

We can now proceed analogously
to~\cite{Mookherjee+Reichelstein:augmented-revelation} to show how the
Augmented Revelation Principle can be used to obtain a necessary
condition for strong implementability of a social choice
function. Mookherjee and Reichelstein already proved that this
condition is sufficient for mechanisms \emph{with} payments, so we
will in fact obtain a necessary \emph{and} sufficient condition in our
setting. Note, however, that the sufficiency result does \emph{not}
hold for mechanisms \emph{without} payments.
To formulate our necessary and sufficient condition for strong
implementability, we need the following definitions:

\begin{definition}
  An equilibrium~$\alpha=(\alpha_1,\dots,\alpha_n)$ in a direct
  revelation mechanism~$\Gamma_{(f,P)}$ can be \emph{selectively
  eliminated} if there exist an agent~$i\in N$ and
  functions~$h:\Theta_{-i} \rightarrow X$ and $\bar{P}_i:\Theta_{-i}
  \rightarrow \mathbb{Q}$ such that:
\begin{enumerate}
\item for some $\bar{\theta}_i\in\Theta_i$:
      $\sum_{\theta_{-i}\in\Theta_{-i}}{q_i(\theta_{-i}|\bar{\theta}_i)\cdot\Big(V_i(h(\alpha_{-i}(\theta_{-i})),\theta_{-i},\bar{\theta}_i)
      +\bar{P}_i(\alpha_{-i}(\theta_{-i}))\Big)}$\newline
    	\phantom{for some $\bar{\theta}_i\in\Theta_i$:}
    $> \sum_{\theta_{-i}\in\Theta_{-i}}{q_i(\theta_{-i}|\bar{\theta}_i)\cdot
      \Big(V_i(f(\alpha_{-i}(\theta_{-i}),\alpha_i(\bar{\theta}_i)),\theta_{-i},\bar{\theta}_i)
      +P_i(\alpha_{-i}(\theta_{-i}),\alpha_i(\bar{\theta}_i))\Big)}$
\item for all $\theta_i\in\Theta_i$: $\sum_{\theta_{-i}\in\Theta_{-i}}{q_i(\theta_{-i}|\theta_i)\cdot\Big(V_i(f(\theta),\theta)+P_i(\theta)
      -V_i(h(\theta_{-i}),\theta)-\bar{P}_i(\theta_{-i})\Big)} \geq 0$
\end{enumerate}
\end{definition}

In this definition, agent~$i$ is given a new bid, which we will call a
\emph{flag}. When agent~$i$ bids the flag and the other agents bid a
vector~$\theta_{-i}\in\Theta_{-i}$, the mechanism chooses the
outcome~$h(\theta_{-i})$ and hands out the payment
$\bar{P}_i(\theta_{-i})$ to agent~$i$. The first condition says that,
for some type~$\bar{\theta}_i\in\Theta_i$, agent~$i$ can increase her
expected utility by deviating from $\alpha_i$ to the flag. Thus,
$\alpha$ is not an equilibrium anymore. However, the second condition
ensures that agent~$i$ can not increase her expected utility by
deviating from truthful reporting to the flag in the case that all
other agents bid their true types. Hence, truthful reporting is
preserved as an equilibrium.

\begin{definition}
  An equilibrium~$\alpha$ of the direct revelation
  mechanism~$\Gamma_{(f,P)}$ is called \emph{bad} if $f\circ\alpha\neq
  f$. The direct revelation mechanism~$\Gamma_{(f,P)}$ satisfies the
  \emph{selective elimination condition} if every bad
  equilibrium~$\alpha$ can be selectively eliminated.
\end{definition}

\begin{theorem}\label{condition}
  Suppose that the social choice function~$f:\Theta \rightarrow X$ is
  strongly implementable. Then there exists an incentive compatible
  direct revelation mechanism~$\Gamma_{(f,P)}$ that satisfies the
  selective elimination condition.
\end{theorem}

\begin{proof}
  By Theorem~\ref{augrefprinciple}, there exists an augmented
  revelation mechanism~$\Gamma=(S_1,\dots,S_n,g,P)$ that strongly
  implements $f$ and in which truthful reporting is an equilibrium. In
  particular, this implies that $g_{|\Theta}=f$.  We now show that the
  direct revelation mechanism~$\Gamma_{(f,P_{|\Theta})}$ is incentive
  compatible and satisfies the selective elimination
  condition. Incentive compatibility follows directly from the fact
  that truthful reporting is an equilibrium in $\Gamma$. Moreover, any
  bad equilibrium~$\alpha$ of $\Gamma_{(f,P_{|\Theta})}$ can not be an
  equilibrium of $\Gamma$ since this would contradict the fact that
  $\Gamma$ strongly implements $f$. Hence, in $\Gamma$ there must be a
  non-type message available to some agent~$i$ to which $i$ prefers to
  deviate when $\alpha$ is being played, without being tempted to do
  the same when all agents report truthfully.  Thus, any bad
  equilibrium~$\alpha$ can be selectively eliminated.
\end{proof}

Mookherjee and Reichelstein~\cite{Mookherjee+Reichelstein:augmented-revelation}
already proved that the condition from Theorem~\ref{condition} is also
sufficient for strong implementability of a social choice function~$f$
in settings where payments are allowed. The idea of the proof is
to start with an incentive compatible direct revelation mechanism and
eliminate the finitely many bad equilibria one after another in order
to obtain an augmented revelation mechanism that strongly
implements~$f$. However, one has to make sure that the augmentations
do not induce new (bad) equilibria. This is done by giving additional
bids (called \emph{counterflags}) to some agent~$j\neq i$ to
make sure that the newly introduced flag of agent~$i$ can never be
used in an equilibrium. To achieve this, the payment scheme has to be
modified appropriately.
Together with Theorem~\ref{condition}, the result of Mookherjee and
Reichelstein~\cite{Mookherjee+Reichelstein:augmented-revelation} proves:

\begin{theorem}\label{characterization}
  The social choice function~$f:\Theta \rightarrow X$ is strongly
  implementable if and only if there exists an incentive compatible
  direct revelation mechanism~$\Gamma_{(f,P)}$ that satisfies the
  selective elimination condition.
\end{theorem} 

Theorem~\ref{characterization} is one of the keys to proving that the
Strong Implementability Problem is in \textsf{PSPACE} (in fact, our
proof shows that the problem is in \textsf{NPSPACE}, which equals
\textsf{PSPACE} by Savitch's Theorem). The idea of our proof is to use
the direct revelation mechanism~$\Gamma_{(f,P)}$ and, for every bad
equilibrium~$\alpha$ of $\Gamma_{(f,P)}$, the index~$i$, the pair of
functions~$(h,\bar{P}_i)$, and the type~$\bar{\theta}_i$ needed to
selectively eliminate $\alpha$, as a certificate for showing that $f$
is strongly implementable.  The algorithm first guesses the
polynomially many values~$P_i(\theta)$ and then enumerates all
possible strategy profiles~$\alpha$ to check which are good or bad
equilibria. If a bad equilibrium~$\alpha$ is found, the algorithm
guesses the data~$(i,h,\bar{P}_i,\bar{\theta}_i)$ needed to
selectively eliminate $\alpha$. However, in order to be able to run
this algorithm in polynomial space, we have to prove that the
certificates used in each step can be chosen to have only polynomial
encoding length.  In particular, we need to show that every
value~$P_i(\theta)$ of the payment functions and every
value~$\bar{P}_i(\theta_{-i})$ can be chosen to have polynomial
encoding length.

\begin{theorem}\label{polymechanism}
  The social choice function~$f:\Theta \rightarrow X$ is strongly
  implementable if and only if there exists an incentive compatible
  direct revelation mechanism~$\Gamma_{(f,P)}$ \emph{of polynomial
    encoding length} that satisfies the selective elimination
  condition.  In this case, for every (fixed) bad
  equilibrium~$\alpha$, the data~$(i,h,\bar{P}_i,\bar{\theta}_i)$
  needed to selectively eliminate $\alpha$ can be chosen to have
  polynomial encoding length.
\end{theorem}

Note that Theorem~\ref{characterization} implies that, in order to
prove the if and only if statement in the claim, we just have to prove
that strong implementability of a social choice function~$f:\Theta
\rightarrow X$ implies the existence of a direct revelation mechanism
with the given properties.

So assume that the social choice function~$f$ is strongly
implementable. Then, by Theorem~\ref{characterization}, there exists
an incentive compatible direct revelation mechanism~$\Gamma_{(f,P)}$
(of possibly more than polynomial encoding length) that satisfies the
selective elimination condition. For each bad equilibrium~$\alpha$ of
$\Gamma_{(f,P)}$, we denote an index~$i$, functions~$h:\Theta_{-i}
\rightarrow X$ and $\bar{P}_i:\Theta_{-i} \rightarrow \mathbb{Q}$, and
a type~$\bar{\theta}_i\in\Theta_i$ that can be used to selectively
eliminate $\alpha$ by
$i_{\alpha},h^{\alpha},\bar{P}^{\alpha}_{i_{\alpha}}$, and
$\bar{\theta}^{\alpha}_{i_{\alpha}}$, respectively. Similarly, for
every strategy profile~$\alpha$ that is \emph{not} an equilibrium, we
denote an index~$i$ and a pair~$(\theta_i,\bar{\theta}_i)$ of types of
agent~$i$ such that
$U^{\Gamma}_i(\alpha_{-i},\bar{\theta}_i|\theta_i)>U^{\Gamma}_i(\alpha_{-i},\alpha_i(\theta_i)|\theta_i)$
by $i_{\alpha}$ and
$(\theta^{\alpha}_{i_{\alpha}},\bar{\theta}^{\alpha}_{i_{\alpha}})$,
respectively.

Note that the only part of the mechanism~$\Gamma_{(f,P)}$ that could
have more than polynomial encoding length is the payment
scheme~$P:\Theta\rightarrow \mathbb{Q}^n$, and, for every bad
equilibrium~$\alpha$, the only part of the
data~$(i_{\alpha},h^{\alpha},\bar{P}^{\alpha}_{i_{\alpha}},\bar{\theta}^{\alpha}_{i_{\alpha}})$
that could have more than polynomial encoding length is the
function~$\bar{P}^{\alpha}_{i_{\alpha}}:\Theta_{-i_{\alpha}}
\rightarrow \mathbb{Q}$.  Hence, we only have to show that every
value~$P_i(\theta)$ of the payment functions and every
value~$\bar{P}^{\alpha}_{i_{\alpha}}(\theta_{-i_{\alpha}})$ can be
chosen to have polynomial encoding length.

To do so, we assume that we are given $i_{\alpha}, h^{\alpha}$, and
$\bar{\theta}^{\alpha}_{i_{\alpha}}$ for every bad
equilibrium~$\alpha$ of $\Gamma_{(f,P)}$, and $i_{\alpha}$,
$(\theta^{\alpha}_{i_{\alpha}},\bar{\theta}^{\alpha}_{i_{\alpha}})$
for every strategy profile~$\alpha$ that is not an equilibrium and
consider the system of linear inequalities in the
variables~$P_i(\theta)$,
$\bar{P}^{\alpha}_{i_{\alpha}}(\theta_{-i_{\alpha}})$, for
$\theta\in\Theta,\,\theta_{-i_{\alpha}}\in\Theta_{-i_{\alpha}}$,
displayed on page~\pageref{thebigthing}.

\begin{sidewaystable}\normalsize
For all strategy profiles~$\alpha$ that are not equilibria:
\begin{eqnarray}\label{thebigthing}
	\sum_{\theta_{-i_{\alpha}}\in\Theta_{-i_{\alpha}}}q_{i_{\alpha}}(\theta_{-i_{\alpha}}|\theta^{\alpha}_{i_{\alpha}})\cdot
	\biggl(V_{i_{\alpha}}(f(\alpha_{-i_{\alpha}}(\theta_{-i_{\alpha}}),\bar{\theta}^{\alpha}_{i_{\alpha}}),
	\theta_{-i_{\alpha}},\theta^{\alpha}_{i_{\alpha}})
	+P_{i_{\alpha}}(\alpha_{-i_{\alpha}}(\theta_{-i_{\alpha}}),\bar{\theta}^{\alpha}_{i_{\alpha}})\biggr) \nonumber\\
	-\sum_{\theta_{-i_{\alpha}}\in\Theta_{-i_{\alpha}}}q_{i_{\alpha}}(\theta_{-i_{\alpha}}|\theta^{\alpha}_{i_{\alpha}})\cdot
	\biggl(V_{i_{\alpha}}(f(\alpha_{-i_{\alpha}}(\theta_{-i_{\alpha}}),\alpha_{i_{\alpha}}(\theta^{\alpha}_{i_{\alpha}})),
	\theta_{-i_{\alpha}},\theta^{\alpha}_{i_{\alpha}})
	+P_{i_{\alpha}}(\alpha_{-i_{\alpha}}(\theta_{-i_{\alpha}}),\alpha_{i_{\alpha}}(\theta^{\alpha}_{i_{\alpha}}))\biggr)
	& > & 0
\end{eqnarray}
For all equilibria~$\alpha$ and all $i\in N,\;\theta_i,\tilde{\theta}_i\in\Theta_i$:
\begin{eqnarray}
	\sum_{\theta_{-i}\in\Theta_{-i}}q_i(\theta_{-i}|\theta_i)\cdot
	\biggl(V_i(f(\alpha_{-i}(\theta_{-i}),\alpha_i(\theta_i)),\theta_{-i},\theta_i)+P_i(\alpha_{-i}(\theta_{-i}),\alpha_i(\theta_i))\biggr) \nonumber\\
	-\sum_{\theta_{-i}\in\Theta_{-i}}q_i(\theta_{-i}|\theta_i)\cdot
	\biggl(V_i(f(\alpha_{-i}(\theta_{-i}),\tilde{\theta}_i),\theta_{-i},\theta_i)+P_i(\alpha_{-i}(\theta_{-i}),\tilde{\theta}_i)\biggr)
	& \geq & 0
\end{eqnarray}
For all bad equilibria~$\alpha$:
\begin{eqnarray}
	\sum_{\theta_{-i_{\alpha}}\in\Theta_{-i_{\alpha}}}q_{i_{\alpha}}(\theta_{-i_{\alpha}}|\bar{\theta}^{\alpha}_{i_{\alpha}})\cdot
	\biggl(V_{i_{\alpha}}(h^{\alpha}(\alpha_{-i_{\alpha}}(\theta_{-i_{\alpha}})),\theta_{-i_{\alpha}},\bar{\theta}^{\alpha}_{i_{\alpha}})
	+\bar{P}^{\alpha}_{i_{\alpha}}(\alpha_{-i_{\alpha}}(\theta_{-i_{\alpha}}))\biggr) \nonumber\\
	-\sum_{\theta_{-i_{\alpha}}\in\Theta_{-i_{\alpha}}}q_{i_{\alpha}}(\theta_{-i_{\alpha}}|\bar{\theta}^{\alpha}_{i_{\alpha}})\cdot
	\biggl(V_{i_{\alpha}}(f(\alpha_{-i_{\alpha}}(\theta_{-i_{\alpha}}),\alpha_{i_{\alpha}}(\bar{\theta}^{\alpha}_{i_{\alpha}})),
	\theta_{-i_{\alpha}},\bar{\theta}_{i_{\alpha}})
	+P_{i_{\alpha}}(\alpha_{-i_{\alpha}}(\theta_{-i_{\alpha}}),\alpha_{i_{\alpha}}(\bar{\theta}^{\alpha}_{i_{\alpha}}))\biggr)
	&  >  & 0
\end{eqnarray}
For all bad equilibria~$\alpha$ and all $\theta_{i_{\alpha}}\in\Theta_{i_{\alpha}}$:
\begin{eqnarray}
	\sum_{\theta_{-i_{\alpha}}\in\Theta_{-i_{\alpha}}}q_{i_{\alpha}}(\theta_{-i_{\alpha}}|\theta_{i_{\alpha}})\cdot
	\biggl(V_{i_{\alpha}}(f(\theta_{-i_{\alpha}},\theta_{i_{\alpha}}),\theta_{-i_{\alpha}},\theta_{i_{\alpha}})
	+P_{i_{\alpha}}(\theta_{-i_{\alpha}},\theta_{i_{\alpha}})
	-V_{i_{\alpha}}(h^{\alpha}(\theta_{-i_{\alpha}}),\theta_{-i_{\alpha}},\theta_{i_{\alpha}})
	-\bar{P}^{\alpha}_{i_{\alpha}}(\theta_{-i_{\alpha}})\biggr)
	& \geq & 0
\end{eqnarray}
\end{sidewaystable}

Here, the inequalities~(1) and (2) encode exactly which strategy
profiles are equilibria of $\Gamma_{(f,P)}$ and (3), (4) correspond to
conditions~1. and 2. in the definition of selective elimination of an
equilibrium, respectively.  Note that the number of inequalities and
variables of the system is exponential in the size of the input of
Strong Implementability.  However, we know that the system has a
solution given by the values~$P_i(\theta),
\bar{P}^{\alpha}_{i_{\alpha}}(\theta_{-i_{\alpha}})$ given by the
mechanism~$\Gamma_{(f,P)}$ and the
functions~$\bar{P}^{\alpha}_{i_{\alpha}}$ specified by the selective
elimination condition for the bad equilibria of $\Gamma_{(f,P)}$.
Theorem~\ref{polymechanism} now follows immediately if we can prove
the following result:

\begin{proposition}\label{polynomialsolution}
  The system of inequalities has a solution in which each component
  has polynomial encoding length.
\end{proposition}

We let $\mathcal{P}$ denote the polyhedron defined by this system of
inequalities when all strict inequalities are replaced by non-strict
inequalities. In other words, $\mathcal{P}$ is the topological closure
of the set of solutions of the system.  We denote the number of bad
equilibria by~$p\in\mathbb{N}$, the number of strategy profiles that
are not equilibria by~$q\in\mathbb{N}$, and the total number of
equilibria by~$m\in\mathbb{N}$.  Then, by multiplying all inequalities
by $-1$ and rearranging, we can write $\mathcal{P}$ as
\begin{equation*}
  \mathcal{P}=\{(x,y)\in\mathbb{R}^{n\cdot|\Theta|}\times\mathbb{R}^l:Ax+By\leq b\},
\end{equation*}
where the vector~$x\in\mathbb{R}^{n\cdot|\Theta|}$ represents the
$n\cdot|\Theta|$ variables~$P_i(\theta)$ and the
vector~$y\in\mathbb{R}^l$ represents the $%\displaystyle
l:=\sum_{\alpha \textrm{ bad equ.}}{|\Theta_{-i_{\alpha}}|}$
variables~$\bar{P}^{\alpha}_{i_{\alpha}}(\theta_{-i_{\alpha}})$.

$A\in\textrm{Mat}(k\times n\cdot|\Theta|,\mathbb{Q})$ and
$B\in\textrm{Mat}(k\times l,\mathbb{Q})$ are the matrices given by the
coefficients of the variables~$P_i(\theta)$ and
$\bar{P}^{\alpha}_{i_{\alpha}}(\theta_{-i_{\alpha}})$ in the system,
respectively, where
$k:=q+m\cdot\sum_{i=1}^n{|\Theta_i|^2}+p+\sum_{\alpha\textrm{ bad equ.}}{|\Theta_{i_{\alpha}}|}$
is the total number of inequalities in the system. The
vector~$b\in\mathbb{Q}^k$ is given by the constant terms in the
system, which are all (sums of) products of valuations and
probabilities. In particular, each entry of the matrices $A$, $B$, and
of the vector~$b$ is of polynomial encoding length.

On our way to proving Proposition~\ref{polynomialsolution} and
Theorem~\ref{polymechanism}, we use some definitions and results from
polyhedral theory. In particular, we use the fact that the vertex and facet
complexity of a polyhedron are polynomially related to each other (see,
e.g.,~\cite{Groetschel+etal:book}).  The most
common use in our context will be to conclude that if each inequality
in a linear system (with potentially exponentially many inequalities)
has polynomial encoding length, then there exists a solution of polynomial
encoding length. We also use the property that
the encoding length of the sum of $2^{\text{poly}(n)}$ many
values~$x_i\in\mathbb{Q}$ is bounded by a polynomial in $\text{poly}(n)$ and
the encoding sizes of the~$x_i$.

\begin{definition}\label{projectiondef}
  The \emph{projection} of a
  polyhedron~$P\subseteq\mathbb{R}^n\times\mathbb{R}^m$ to
  $\mathbb{R}^n$ is defined as
  \begin{equation*}
    \operatorname{proj}_{\mathbb{R}^n}(P):=\{x\in\mathbb{R}^n:(x,y)\in
    P \text{ for some } y\in\mathbb{R}^m\}.
  \end{equation*}
\end{definition}

By standard results from polyhedral theory (cf. \cite{Nemhauser+Wolsey:book}),
the projection of the polyhedron~$\mathcal{P}$, which is the closure of the set
of solutions of our system of linear inequalities, to $\mathbb{R}^{n\cdot|\Theta|}$
can be written as
\begin{equation*}
  \operatorname{proj}_{\mathbb{R}^{n\cdot|\Theta|}}(\mathcal{P})
  :=\{x\in\mathbb{R}^{n\cdot|\Theta|}:r_{\lambda}^T(b-Ax)\geq0 
  \text{ for all } \lambda\in\Lambda\}, 
\end{equation*}
where $\{r_{\lambda}\}_{\lambda\in   \Lambda}:=\operatorname{extreme.rays}(Q)$
is the finite set of extreme rays of the polyhedron
$Q:=\{v\in\mathbb{R}^k_+:v^T B=0\}$.  In particular,
$\textrm{proj}_{\mathbb{R}^{n\cdot|\Theta|}}(\mathcal{P})$ is a
polyhedron.

\begin{lemma}\label{extremerays}
  Let $P_1\subseteq\mathbb{R}^n$ and $P_2\subseteq\mathbb{R}^m$ be
  polyhedra and $P=P_1\times P_2$.Then
  \begin{equation*}
    \operatorname{extreme.rays}(P)=\operatorname{extreme.rays}(P_1)\times\operatorname{extreme.rays}(P_2).
  \end{equation*}
\end{lemma}

\begin{lemma}\label{relativeinterior}
  Let $K^{\textrm{ri}}$ denote the \emph{relative interior} of a
  convex set~$K$.  If $P\subseteq\mathbb{R}^n\times\mathbb{R}^m$ is a
  polyhedron, then
  $\textrm{proj}_{\mathbb{R}^n}(P)^{\textrm{ri}}=\textrm{proj}_{\mathbb{R}^n}(P^{\textrm{ri}})$.
  Moreover, if $P=P_1\times P_2$ for polyhedra
  $P_1\subseteq\mathbb{R}^n$ and $P_2\subseteq\mathbb{R}^m$, then
  $P^{\textrm{ri}}=P_1^{\textrm{ri}}\times P_2^{\textrm{ri}}$.
\end{lemma}

%\begin{proof}
%For a proof of the second statement, see Proposition~2.2.11 in \cite{Hiriart-Urruty+Lemarechal:book}.
%
%\bigskip
%
%To prove the first statement, we define the matrix~$M\in\textrm{Mat}((n+m)\times (n+m),\mathbb{R})$ as
%\begin{displaymath}
%M:=
%\left(\begin{array}{cc}
%\textrm{I} & 0 \\
%0 & 0
%\end{array}\right),
%\end{displaymath}
%where $\textrm{I}\in\textrm{Mat}(n\times n,\mathbb{R})$ is the $n\times n$-unit matrix.
%By definition of the projection and the fact that the relative interior commutes with linear maps (cf. for
%example~\cite{Hiriart-Urruty+Lemarechal:book}, Proposition 2.1.12), we then have
%\begin{displaymath}
%\textrm{proj}_{\mathbb{R}^n}(P)^{\textrm{ri}}\times\{0\}=(\textrm{proj}_{\mathbb{R}^n}(P)\times\{0\})^{\textrm{ri}}=(M\cdot P)^{\textrm{ri}}
%=M\cdot(P^{\textrm{ri}})=\textrm{proj}_{\mathbb{R}^n}(P^{\textrm{ri}})\times\{0\},
%\end{displaymath}
%which proves the claim.\qed
%\end{proof}

The following proposition is crucial for the proof of
Proposition~\ref{polynomialsolution}.  We use $\operatorname{ker}(C)$
to denote the kernel~$\{ \,z: Cz=0\,\}$ of a matrix~$C$.

\begin{proposition}\label{extremerayspolynomial}
  Each entry of an extreme ray~$r_{\lambda}$ of the polyhedron~$Q$ has
  encoding length polynomial in the input size of the Strong
  Implementability Problem.
\end{proposition}

\begin{proof}
$Q$ can be written as $Q=\{v\in\mathbb{R}^k_+:B^T v=0\}=\textrm{ker}(B^T)\cap\mathbb{R}^k_+$,
where the matrix $B^T$ has the form\small
\begin{equation}\label{eq:1}
B^T=
\left(\begin{array}{c@{\qquad}|ccccc}
& B_1^T & \quad & \quad && \text{\large $0$} \\
& \quad & B_2^T & \quad \\
\quad\text{\Large $0$}& \quad & \quad & \ddots & \quad \\
&  & \quad & \quad & B_{p-1}^T\\
& \text{\large $0$} & \quad & \quad && B_p^T
\end{array}\right) 
= 
\left(\begin{array}{c|c}
0& \tilde{B}
\end{array}\right) 
\end{equation}\normalsize
Here,
$B_j\in\textrm{Mat}((1+|\Theta_{i_{\alpha_j}}|)\times|\Theta_{-i_{\alpha_j}}|,\mathbb{Q})$
is the coefficient matrix of the $|\Theta_{-i_{\alpha_j}}|$
variables~$\bar{P}^{\alpha_j}_{i_{\alpha_j}}(\theta_{-i_{\alpha_j}})$,
$\theta_{-i_{\alpha_j}}\in\Theta_{-i_{\alpha_j}}$, for the bad
equilibrium~$\alpha_j$ in the inequalities (3) and (4) corresponding
to $\alpha_j$. Hence, we can write $Q$ as
\begin{equation*}
  Q=\mathbb{R}^{q+m\cdot\sum_{i=1}^n{|\Theta_i|^2}}_+ \times
  \left(\textrm{ker}(\tilde{B})\cap\mathbb{R}_+^{p+\sum_{\alpha \textrm{ bad equ.}}{|\Theta_{i_{\alpha}}|}}\right)
\end{equation*}
where $\tilde{B}\in\textrm{Mat}(l\times (p+\sum_{\alpha \textrm{ bad
    equ.}}{|\Theta_{i_{\alpha}}|}),\mathbb{Q})$ is the submatrix of
$B^T$ obtained by deleting the $q+m\cdot\sum_{i=1}^n{|\Theta_i|^2}$
zeros at the beginning of each line as shown in~(\ref{eq:1}).
Hence, by Lemma~\ref{extremerays}, every extreme ray~$r_{\lambda}$ of
$Q$ is of the form~$r_{\lambda}=(e_i,r)$, where $e_i$ is a unit vector
in $\mathbb{R}^{q+m\cdot\sum_{i=1}^n{|\Theta_i|^2}}$ (which are the
extreme rays of $\mathbb{R}^{q+m\cdot\sum_{i=1}^n{|\Theta_i|^2}}_+$)
and $r$ is an extreme ray of
$\textrm{ker}(\tilde{B})\cap\mathbb{R}_+^{p+\sum_{\alpha \textrm{ bad
      equ.}}{|\Theta_{i_{\alpha}}|}}$.

To obtain a description of the extreme rays of
$\textrm{ker}(\tilde{B})\cap\mathbb{R}_+^{p+\sum_{\alpha \textrm{ bad equ.}}{|\Theta_{i_{\alpha}}|}}$, note that $\tilde{B}$ is the
direct sum of the matrices~$B_1^T,\dots,B_p^T$, so we have
\begin{equation*}
  \textrm{ker}(\tilde{B})\cap\mathbb{R}_+^{p+\sum_{\alpha \textrm{ bad equ.}}{|\Theta_{i_{\alpha}}|}}=
  \bigoplus_{j=1}^p\left(\textrm{ker}(B_j^T)\cap\mathbb{R}^{1+|\Theta_{i_{\alpha_j}}|}_+\right).
\end{equation*}
Thus, again by Lemma~\ref{extremerays}, every extreme ray~$r_{\lambda}$ of $Q$ is of the form~$r_{\lambda}=(e_i,r_1,\dots,r_p)$, where
$e_i$ is a unit vector in $\mathbb{R}^{q+m\cdot\sum_{i=1}^n{|\Theta_i|^2}}$ and
$r_j\in\mathbb{R}^{1+|\Theta_{i_{\alpha_j}}|}$ is an extreme ray of $Q_j:=\textrm{ker}(B_j^T)\cap\mathbb{R}^{1+|\Theta_{i_{\alpha_j}}|}_+$
for every $j=1,\dots,p$. Hence, it just remains to show that, for every $j=1,\dots,p$, each entry of an extreme ray~$r_j$ of $Q_j$
has polynomial encoding length. This follows by writing $Q_j=\{v\in\mathbb{R}^{1+|\Theta_{i_{\alpha_j}}|}_+: B_j^T v=0\}$
and the fact that the encoding length of $B_j$ is polynomial for every $j$.
\end{proof}

\begin{proposition}\label{projectioncomplexity}
The facet complexity of $\textrm{proj}_{\mathbb{R}^{n\cdot|\Theta|}}(\mathcal{P})$ is polynomial in the encoding length of the input
of Strong Implementability.
\end{proposition}

\begin{proof}
As already stated, we have
\begin{eqnarray}
\textrm{proj}_{\mathbb{R}^{n\cdot|\Theta|}}(\mathcal{P})
& = & \{x\in\mathbb{R}^{n\cdot|\Theta|}:r_{\lambda}^T(b-Ax)\geq0 \textrm{ for all } \lambda\in \Lambda\} \nonumber\\
& = & \{x\in\mathbb{R}^{n\cdot|\Theta|}:(r_{\lambda}^TA)x\leq r_{\lambda}^T b \textrm{ for all } \lambda\in \Lambda\} \nonumber
\end{eqnarray}
where $\operatorname{extreme.rays}(Q)=\{r_{\lambda}\}_{\lambda\in
  \Lambda}$ is the finite set of extreme rays of $Q$. Hence, the claim
follows if we show that each inequality $(r_{\lambda}^TA)x\leq
r_{\lambda}^T b$ is of polynomial encoding length. To this end,
consider the inequality~$(r_{\lambda}^TA)x\leq r_{\lambda}^T b$ for a
fixed $\lambda\in \Lambda$. By
Proposition~\ref{extremerayspolynomial}, each entry of the extreme ray
$r_{\lambda}$ has polynomial encoding length. Since each entry of~$b$
has polynomial encoding length as well, the value $r_{\lambda}^T
b\in\mathbb{Q}$ is a sum of exponentially many values of polynomial
encoding length, which is again of polynomial encoding
length. Similarly, since each entry of the matrix~$A$ has polynomial
encoding length, each entry of the vector~$r_{\lambda}^T A$ has
polynomial encoding length, which implies that the whole
vector~$r_{\lambda}^T A$ is of polynomial encoding length since the
vector is of polynomial size. Thus, the encoding length of
$(r_{\lambda}^TA)x\leq r_{\lambda}^T b$ is polynomial.
\end{proof}

%%Now we are ready to prove Proposition~\ref{polynomialsolution}.

\begin{proof}[Proof of Proposition~\ref{polynomialsolution}]
Let $\textrm{eq}(\mathcal{P})$ (the \emph{equality set} of $\mathcal{P}$) denote the set of indices of inequalities in $Ax+By\leq b$ that
are satisfied with equality for all points in $\mathcal{P}$.
It is enough to prove the existence of a point~$(\bar{x},\bar{y})$ in $\mathcal{P}^{\textrm{ri}}$ such that each component
of $(\bar{x},\bar{y})$ has polynomial encoding length: Such a point~$(\bar{x},\bar{y})$ has to
satisfy all inequalities with indices \emph{not} in $\textrm{eq}(\mathcal{P})$ with strict inequality. 
Moreover, since the original system (where the inequalities (1) and (3) are strict) has a solution, we know that all indices of the
inequalities~(1) and (3) are not in $\textrm{eq}(\mathcal{P})$. Hence, $(\bar{x},\bar{y})$ satisfies all the inequalities~(1) and (3)
with strict inequality, i.e., it is a solution of the original system.

By Proposition~\ref{projectioncomplexity}, $\textrm{proj}_{\mathbb{R}^{n\cdot|\Theta|}}(\mathcal{P})$ is a nonempty polyhedron
of polynomial facet complexity. Using standard results from polyhedral theory  (cf. \cite[Thm. 6.5.5]{Groetschel+etal:book}),
this implies the existence of a point~$\bar{x}$ of polynomial encoding length in
$\textrm{proj}_{\mathbb{R}^{n\cdot|\Theta|}}(\mathcal{P})^{\textrm{ri}}$. In particular, each component of $\bar{x}$ has polynomial encoding length.
Moreover, by Lemma~\ref{relativeinterior}, we have $\bar{x}\in\textrm{proj}_{\mathbb{R}^{n\cdot|\Theta|}}(\mathcal{P}^{\textrm{ri}})$, so we can
choose $y_0\in\mathbb{R}^l$ such that $(\bar{x},y_0)\in \mathcal{P}^{\textrm{ri}}$, i.e., $y_0$ is a solution of the system
$A\bar{x}+By\leq b \Leftrightarrow By\leq b-A\bar{x}$ and all inequalities with indices not in $\textrm{eq}(\mathcal{P})$
are satisfied strictly for $y=y_0$. Writing $\tilde{\mathcal{P}}:=\{y\in\mathbb{R}^l: By\leq b-A\bar{x}\}$,
we have $\textrm{eq}(\tilde{\mathcal{P}})=\textrm{eq}(\mathcal{P})$:
The inclusion~$\textrm{eq}(\tilde{\mathcal{P}})\supseteq\textrm{eq}(\mathcal{P})$ follows from the definition of $\tilde{\mathcal{P}}$ and
the other inclusion follows since $y_0\in\tilde{\mathcal{P}}$ satisfies all inequalities with indices not in $\textrm{eq}(\mathcal{P})$
with strict inequality. The matrix~$B$ has the form\small
\begin{displaymath}
B=
\left(\begin{array}{cccc}
\,0\, & \,0\, & \,0\, & \,0\, \\
B_1 & \quad & \quad & 0 \\
\quad & B_2 & \quad \\
\quad & \quad & \ddots & \quad \\
0 & \quad & \quad & B_p
\end{array}\right).
\end{displaymath}\normalsize
Hence, the system~$By\leq b-A\bar{x}$ decomposes into $p$ smaller systems
$B_jy^j\leq b^j-A^j\bar{x}^j$, $j=1,\dots,p$, where $b^j,A^j,\bar{x}^j$ denote the parts of $b,A$ and $\bar{x}$,
respectively, corresponding to the lines of the system containing the submatrix~$B_j$.
Writing
\begin{displaymath}
\tilde{\mathcal{P}}_j:=\{y^j\in\mathbb{R}^{|\Theta_{-i_{\alpha_j}}|}:B_jy^j\leq b^j-A^j\bar{x}^j\},
\end{displaymath}
we have $\tilde{\mathcal{P}}=\tilde{\mathcal{P}}_1\times\dots\times\tilde{\mathcal{P}}_p$,
and each $\tilde{\mathcal{P}}_j$ is nonempty because $\tilde{\mathcal{P}}$ is nonempty.

Since $B_j\in\textrm{Mat}((1+|\Theta_{i_{\alpha_j}}|)\times|\Theta_{-i_{\alpha_j}}|)$ is of polynomial size for every $j$,
and each entry of $A,B,b$, and $\bar{x}$ has polynomial encoding length, the facet complexity of each polyhedron~$\tilde{\mathcal{P}}_j$
is polynomial. Thus, again by standard results from polyhedral theory, the relative interior of each $\tilde{\mathcal{P}}_j$ contains a
point~$\tilde{y}^j$ of polynomial encoding length.

If we now define $\bar{y}:=(\tilde{y}^1,\dots,\tilde{y}^p)$, all components of the vector~$(\bar{x},\bar{y})$ have polynomial encoding length
and we have $\bar{y}\in\tilde{\mathcal{P}}_1^{\textrm{ri}}\times\dots\times\tilde{\mathcal{P}}_p^{\textrm{ri}}=\tilde{\mathcal{P}}^{\textrm{ri}}$,
where the equality follows by Lemma~\ref{relativeinterior}.
Hence, $(\bar{x},\bar{y})$ satisfies $A\bar{x}+B\bar{y}\leq b$ and all inequalities with indices not in
$\textrm{eq}(\tilde{\mathcal{P}})=\textrm{eq}(\mathcal{P})$ are satisfied with strict inequality, i.e.,
$(\bar{x},\bar{y})\in\mathcal{P}^{\textrm{ri}}$.
\end{proof}

\begin{theorem}
Strong Implementability $\in$ \textsf{PSPACE}.
\end{theorem}

\begin{proof}
Assume that the given social choice function~$f$ is strongly implementable. Then, by Theorem~\ref{polymechanism}, there exists an
incentive compatible direct revelation mechanism~$\Gamma_{(f,P)}$ of polynomial encoding length that satisfies the selective
elimination condition. Moreover, for every bad equilibrium~$\alpha$ of $\Gamma_{(f,P)}$, the
data~$(i_{\alpha},h^{\alpha},\bar{P}^{\alpha}_{i_{\alpha}},\bar{\theta}^{\alpha}_{i_{\alpha}})$ needed to selectively eliminate $\alpha$
can also be chosen to have polynomial encoding length. Now consider Algorithm~\ref{polyspacealg} on the following page for verifying
that $f$ is strongly implementable.

Since all the values~$P_i(\theta)$ and the data~$(i_{\alpha},h^{\alpha},\bar{P}^{\alpha}_{i_{\alpha}},\bar{\theta}^{\alpha}_{i_{\alpha}})$ for
every bad equilibrium~$\alpha$ have polynomial encoding length and every inequality in the system is of polynomial encoding length,
Algorithm~\ref{polyspacealg} uses only polynomial space, which proves the claim.\qedhere

\begin{algorithm}\quad\\\label{polyspacealg}
	\begin{tabularx}{\linewidth}{lX}
	1. & Guess the (polynomially many) values~$P_i(\theta)$. \\
	2. & For every strategy profile~$\alpha=(\alpha_1,\dots\alpha_n)$ in the mechanism~$\Gamma_{(f,P)}$ do:\vspace{-2mm}
				\begin{itemize}
				\item Check whether $\alpha$ is an equilibrium by going through the inequalities~(2) corresponding to $\alpha$ one by one.
				\item If $\alpha$ is \emph{not} an equilibrium, we have already found
							$(i_{\alpha},\theta^{\alpha}_{i_{\alpha}},\bar{\theta}^{\alpha}_{i_{\alpha}})$ such that the inequality~(1) corresponding
							to $\alpha$ is satisfied in the previous step.
				\item If $\alpha$ is an equilibrium, check whether $f\circ\alpha=f$ by going through all possible bid vectors~$\theta\in\Theta$.
							If $f\circ\alpha\neq f$, guess the data~$(i_{\alpha},h^{\alpha},\bar{P}^{\alpha}_{i_{\alpha}},\bar{\theta}^{\alpha}_{i_{\alpha}})$
							of polynomial encoding length needed to selectively eliminate $\alpha$ and check the inequalities~(3) and (4) corresponding
							to $\alpha$ one by one.
				\end{itemize}
	\end{tabularx}\vspace{-5mm}
\end{algorithm}
\end{proof}

As a byproduct of Theorem~\ref{polymechanism}, we obtain the following result on the size of the payments needed to strongly implement a
social choice function:

\begin{corollary}\label{polypayments}
If the social choice function~$f:\Theta \rightarrow X$ is strongly implementable, it can be strongly implemented by an augmented revelation
mechanism~$\Gamma=(S_1,\dots,S_n,g,P)$ in which each payment~$P_i(s)$, $s\in S$, $i\in N$, has polynomial encoding length.
\end{corollary}

\begin{proof}
If $f$ is strongly implementable, Theorem~\ref{polymechanism} implies the existence of an incentive compatible direct revelation
mechanism~$\Gamma^{\prime}_{(f,P^{\prime})}$ of polynomial encoding length that satisfies the selective elimination condition and, for every
bad equilibrium~$\alpha$ of $\Gamma^{\prime}_{(f,P^{\prime})}$, the
data~$(i_{\alpha},h^{\alpha},\bar{P}^{\alpha}_{i_{\alpha}},\bar{\theta}^{\alpha}_{i_{\alpha}})$ needed to selectively eliminate $\alpha$
can also be chosen to have polynomial encoding length. In particular, all payments and ``elimination payments''~$\bar{P}^{\alpha}_{i_{\alpha}}$
are of polynomial encoding length. As shown in the proof of the sufficiency part of Theorem~\ref{characterization} given in
\cite{Mookherjee+Reichelstein:augmented-revelation}, we can then obtain an augmented revelation mechanism as in the claim from
$\Gamma^{\prime}_{(f,P^{\prime})}$ by a sequence of at most $|\Theta|^{|\Theta|}\leq 2^{|\Theta|^2}$ augmentations.
In each step, one of the at most $|\Theta|^{|\Theta|}$ bad equilibria is eliminated without introducing any new equilibria.
The payments after each augmentation are given by the $\bar{P}^{\alpha}_{i_{\alpha}}$ and the payments from the previous step, which are
only changed by additive terms of polynomial encoding length. Since there are at most $|\Theta|^{|\Theta|}\leq 2^{|\Theta|^2}$ augmentation steps
and each of the initial payments in $\Gamma^{\prime}_{(f,P^{\prime})}$ is also of polynomial encoding length, this implies that each
payment~$P_i(s)$ in the final augmented revelation mechanism that strongly implements $f$ is of polynomial encoding length.
\end{proof}

\section{The Single-Agent Scenario}\label{section-singleagent}

In this section, we show that, in the case of a single agent, Strong Implementability can be decided in polynomial time via
linear programming.

%In the single-agent scenario, we denote the set of possible types of the agent by $\Theta$ and her valuation function by
%$V:X\times\Theta \rightarrow \mathbb{Q}$. The agent's set of possible bids in a mechanism is typically denoted by $S$.
%The formal definition of the utility function of the agent and an equilibrium strategy in this setting is the following:

\begin{definition}
Given a single-agent mechanism~$\Gamma=(S,g,P)$ and a type~$\theta\in\Theta$ of the agent, the \emph{utility} from a bid~$s\in S$ for the agent
is defined as $U^{\Gamma}(s|\theta):= V(g(s),\theta)+P(s)$.
A strategy~$\alpha:\Theta \rightarrow S$ of the agent in the mechanism~$\Gamma$ is an
\emph{equilibrium strategy} (or simply an \emph{equilibrium}) if $\alpha(\theta)$ maximizes the utility of the agent for every
$\theta\in\Theta$, i.e., if
$U^{\Gamma}(\alpha(\theta)|\theta)\geq U^{\Gamma}(\bar{s}|\theta)$ for all $\theta\in\Theta,\bar{s}\in S$.
\end{definition}

Note that this definition is just the special case of
Definition~\ref{bayesianequ} where there is only a single agent.
Incentive compatibility of a direct revelation mechanism and strong
implementation of a social choice function are defined as before. The
Strong Implementability Problem for a single agent (Single-Agent Strong
Implementability) is simply the special case of the Strong Implementability
Problem with the number~$n$ of agents equal to one. Note that the probability
distribution~$p$ is not needed for a single agent.

%% Formally, it is defined as follows:

% \begin{definition}[Single-Agent Strong Implementability Problem]\mbox{}\\
%   \begin{tabularx}{\linewidth}{lX}
%     INSTANCE: & The set~$X$ of possible social choices, the set~$\Theta$ of possible types of the agent,
%     the valuation function~$V:X\times\Theta \rightarrow \mathbb{Q}$, and the social choice function~$f:\Theta \rightarrow X$.\\
%     QUESTION: & Is $f$ strongly implementable? \
%   \end{tabularx}
% \end{definition}

As Single-Agent Strong Implementability is a special case of Strong Implementability, the characterization
from Theorem~\ref{characterization} holds in the single-agent scenario as well. However, the definition of
selective elimination implies that, with just one agent, selective elimination of a bad equilibrium is \emph{never}
possible. Consequently, a single-agent direct revelation mechanism~$\Gamma_{(f,P)}$ satisfies the selective elimination condition
if and only if it has no bad equilibria at all and we obtain the following result:

\begin{theorem}\label{singleagentcharacterization}
A social choice function~$f:\Theta \rightarrow X$ in the single-agent scenario is strongly implementable if and only if there exists an incentive
compatible direct revelation mechanism~$\Gamma_{(f,P)}$ without bad equilibria.
\end{theorem}

We will now show how we can use Theorem~\ref{singleagentcharacterization} to decide Single-Agent Strong Implementability in polynomial time.
In the case of a single agent, a bad equilibrium in a direct revelation mechanism is simply a strategy~$\alpha:\Theta \rightarrow \Theta$ of
the agent such that $\alpha(\theta)$ maximizes the utility of the agent for every $\theta\in\Theta$ and such that $f(\theta)\neq f(\alpha(\theta))$
for at least one $\theta\in\Theta$.
Hence, an incentive compatible direct revelation mechanism~$\Gamma_{(f,P)}$ has no bad equilibrium if and only if we have
\begin{displaymath}
V(f(\theta^{\prime}),\theta)+P(\theta^{\prime}) < V(f(\theta),\theta)+P(\theta)
\end{displaymath}
for all $\theta,\theta^{\prime}\in\Theta$ with $f(\theta)\neq f(\theta^{\prime})$. This implies that
the possible payment schemes~$P$ of an incentive compatible direct revelation mechanism~$\Gamma_{(f,P)}$ satisfying the selective
elimination condition correspond to the solutions of the following system in the variables~$P(\theta)$ for $\theta\in\Theta$:
\begin{eqnarray}
V(f(\theta^{\prime}),\theta)+P(\theta^{\prime}) & < & V(f(\theta),\theta)+P(\theta) \quad\forall \theta,\theta^{\prime}\in\Theta
\textrm{ with } f(\theta)\neq f(\theta^{\prime}) \label{strictieq}\\
V(f(\bar{\theta}),\theta)+P(\bar{\theta}) & \leq & V(f(\theta),\theta)+P(\theta) \quad\forall \theta,\bar{\theta}\in\Theta \label{incentineq}
\end{eqnarray}
Here, the inequalities~(\ref{incentineq}) encode incentive compatibility of the mechanism and, as discussed above, the strict
inequalities~(\ref{strictieq}) encode that the mechanism has no bad equilibrium.
The polyhedron~$\mathcal{P}$ defined by this system when all strict inequalities are replaced by non-strict inequalities has polynomial facet
complexity, since the system only contains polynomially many variables (and inequalities) and coefficients of polynomial encoding length. Thus,
we can check in polynomial time whether the polyhedron contains a relative interior point corresponding to a solution of the original system with
strict inequalities in (\ref{strictieq}). Hence, we obtain the following result:

\begin{theorem}
Single-Agent Strong Implementability $\in$ \textsf{P}.
\end{theorem}

In fact, the above system shows that Single-Agent Strong Implementability can be decided via linear programming.

\bibliographystyle{eptcs}
\bibliography{lit_diss_english}
\end{document}